\documentclass{llncs}
\usepackage{llncsdoc}

\usepackage{amssymb}

\usepackage{amsmath}

\usepackage{epsfig}

\usepackage[usenames]{color}

\newcommand{\R}{{\bf R}}

\newcommand{\ab}{\allowbreak}

\newcommand{\ignore}[1]{}

\newcommand{\ord}[1]{\mathcal{O}\!\left(#1\right)}

\newcommand{\out}[1]{}

\def\eu{\,\hbox{\raise .36em\hbox to0pt{\vrule height0.5pt%
width.55em depth0pt\hss}%
\raise .25em\hbox to0pt{\vrule height0.5pt width.5em%
depth0pt\hss}\hskip.02em\sf C}\,}

\begin{document}
\pagenumbering{arabic} \thispagestyle{empty}
 \title{Aggregation Languages for Moving Object and Places of Interest Data}

\author{Leticia G\'omez \inst{1,3} \and Bart Kuijpers \inst{2} \and  Alejandro Vaisman \inst{3}}
\mainmatter

\institute{Instituto Tecnol\'ogico de Buenos Aires\\
\email{lgomez@itba.edu.ar} \and
Hasselt University and Transnational University of Limburg, Belgium\\
\email{bart.kuijpers@uhasselt.be} \and Universidad de Buenos
Aires\\
\email{avaisman@dc.uba.ar}}

\maketitle

\begin{abstract}
 We address aggregate queries over GIS data and moving object data,
 where non-spatial data are stored in a data warehouse.
We propose a formal data model and query language to express complex aggregate
queries. Next, we study the compression of trajectory data, produced by moving objects,
using the notions of stops and moves. We show that stops and moves are expressible
in our query language and we consider a fragment of this language, consisting of
 regular expressions to talk about temporally ordered sequences of stops and moves.
 This fragment can be used to efficiently
 express data mining and pattern matching tasks over
trajectory data.
\end{abstract}



\section{Introduction}
\label{sec:intro}

Geographic Information Systems (GIS) have been
extensively used in various application domains, ranging from
economical, ecological and demographic analysis, to city and route
planning~\cite{Rigaux02,worboys}.
In recent years, \emph{time} is playing an increasingly important role in
in GIS and spatial data management~\cite{timeingis}.
  One particular line of research in this direction, introduced by
  Wolfson~\cite{Guting00,Guting05,Meratnia02,Wolfson99,Wolfson98,Vazirgiannis01},
   concerns \emph{moving object data}.
Moving objects, carrying location-aware devices, produce trajectory
data in the form of a sample of $(O_{id},t,x,y)$-tuples, that
 contain object identifier and time-space information.

In this paper, we are interested in \emph{aggregate queries}
over GIS data and moving object data.
Typically, when aggregation  becomes important, it is
advisable to organize the non-spatial data in a GIS in a data warehouse.
In a data warehouse, numerical data are stored in fact tables
built along several dimensions. For instance, if we are interested in
the sales of certain products in stores in a given region, we may
consider the sales amounts in a fact table over the three dimensions
store, time and product.  In general dimensions are organized into
 aggregation hierarchies. For example, stores can aggregate over
  cities which in turn can aggregate into regions and countries.
  Each of these aggregation levels can also hold descriptive attributes
  like city population, the area of a region, etc.
For traditional alpha-numeric data,
  OLAP (On Line Analytical Processing)~\cite{Kimball02}
comprises a set of tools  and algorithms that allow efficiently
querying multidimensional
    databases, containing large amounts of data, usually called
    data warehouses.

 Two of the present authors have proposed in previous work a model
 for smoothly integrating  the GIS and OLAP  worlds. This model was
 implemented using open source software~\cite{Haesevoets06}.
 The same authors also proposed a taxonomy of aggregation queries on moving
 object data~\cite{Kuijpers07}.
 In this paper, we propose a conceptual model and a formal query language
 that cover the different types of aggregation queries discussed in the above
  mentioned taxonomy (see Sections~\ref{sec:DataModel} and~\ref{sec:FOlang}).
 At the basis of our aggregation query language is a multi-sorted first-order
 query language ${\cal L}_{mo}$ for moving object and GIS data in
 which one can specify properties of moving objects, geometric
 elements of GIS layers and OLAP data storing the non-spatial GIS data.

 Recently, in the study of moving object data, the concepts
  of \emph{stops} and \emph{moves} were introduced~\cite{Mouza05,Damiani07}.
 These concepts serve to compress the trajectory data that
 is produced by moving objects using application dependent
  places of interest.  A designer may want to select a set
   of places of interest that are relevant to her application. For
   instance, in a tourist application, such places can be hotels,
   museums and churches. In a traffic control application, they may
   be road segments, traffic lights and junctions. We assume that these
   places of interest are stored in a specific GIS layer. If a moving
   object spends a sufficient amount of time in a place of interest,
   this place is considered a stop of the object's trajectory. In
   between stops, the trajectory has moves. Thus, we can replace a
 raw trajectory given by $(O_{id},t,x,y)$-tuples by a sequence of
 application-relevant stops and moves.
 In this paper, we give a geometric definition of stops and
 moves and show that they are computable (see Section~\ref{sec:trajectories}).
We also show that this compression can be expressed in the language
   ${\cal L}_{mo}$ and we sketch a sublanguage of  ${\cal L}_{mo}$ that
    allows us to talk about temporally ordered sequences of stops and
     moves (see Section~\ref{sec:stopsmoves}). The syntax of this
     languages is given in the form of regular expressions
     (see Section~\ref{sec:language}). We show that this language
      considerably  extends the language proposed by Mouza and
       Rigaux~\cite{Mouza05}, and can be used
       to efficiently express data mining and pattern matching tasks over
trajectory data.

  \subsection{Running Example}

 Now, let us introduce the example we will be using throughout the
paper. Figure \ref{fig:paris}  (left) shows a simplified map of Paris,
containing two hotels, denoted Hotel 1 and Hotel 2 (H1 and H2 from
here on), the Louvre and the Eiffel tower. We consider three
moving objects, O1, O2 and O3. Object O1 goes from H1 to the
Louvre, the Eiffel tower, spends just a few minutes there,  and
returns to the hotel.
 Object O2 goes from H1 to the Louvre, the Eiffel tower
 (it stays a couple of hours in each place),  and returns to the hotel.
 Object O3 leaves H2 to the Eiffel tower, visits the place, and
 returns to H2. Figure~\ref{fig:paris}  shows an example of these
 trajectory samples on the right.

 In this scenario, a GIS user may be interested in finding out useful trajectory
 information in this setting, like ``number of persons
  going from H1 to the Louvre and the
 Eiffel tower (visiting both places) in the same day'',
  or  ``number of persons going from a hotel in the left bank of the
 Seine, to the Louvre in the mornings''.

\begin{figure}[t]
\centerline{
\begin{tabular}{l c r}
   \raisebox{-2cm}{\psfig{figure=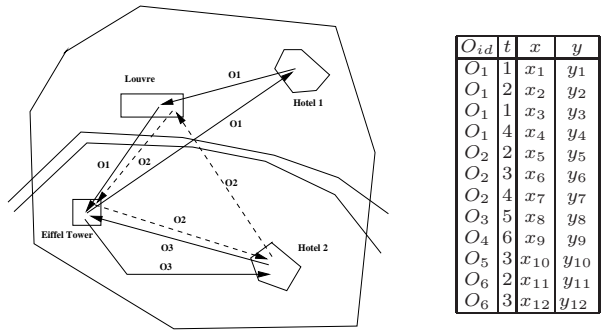,width=2in,height=1.7in}}
   &\phantom{joske}&{\scriptsize \begin{tabular}{|c|c|c|c|}
        \hline
      $O_{id}$ & $t$  & $x$ & $y$\\\hline
   $O_1$ & 1 & $x_1$ & $y_1$  \\
   $O_1$ & 2 & $x_2$ & $y_2$   \\
    $O_1$ & 1 & $x_3$ & $y_3$   \\
      $O_1$ & 4 &$x_4$ & $y_4$   \\
       $O_2$ & 2 & $x_5$ & $y_5$   \\
        $O_2$ & 3 & $x_6$ & $y_6$   \\
         $O_2$ & 4 & $x_7$ & $y_7$  \\
          $O_3$ & 5 &$x_8$ & $y_8$   \\
           $O_4$ & 6 & $x_9$ & $y_9$  \\
            $O_5$ & 3 & $x_{10}$ & $y_{10}$  \\
             $O_6$ & 2 & $x_{11}$ & $y_{11}$   \\
              $O_6$ & 3 & $x_{12}$ & $y_{12}$   \  \\ \hline
        \end{tabular}}
        \end{tabular}
}
\caption{Running example (left) and  a moving object  fact table (right)} \label{fig:paris}
\end{figure}

\subsection{Related Work}
\label{sec:related}



 \paragraph{GIS and OLAP Interaction}
   Although some authors have pointed out the
   benefits
   of combining GIS and OLAP, not much work has been done in this
    field. Vega L\'opez {\em et al}~\cite{VegaLopez05} present a
    comprehensive survey on spatiotemporal  aggregation that
    includes a section on spatial aggregation.
         Rivest {\em et
     al.}~\cite{Rivest01} introduce the concept of SOLAP (standing for
     Spatial OLAP), and describe the desirable features and operators  a SOLAP
   system should have.
    Han {\em et al.}~\cite{Han98} used OLAP techniques for materializing
    selected spatial objects, and proposed a so-called {\em
    Spatial Data Cube}.  This model only supports
    aggregation of such  spatial objects.


  \paragraph{Moving Objects}

Many efforts have been made in the field of moving objects
databases, specially regarding data modeling an indexing. G\"uting
and Schneider~\cite{Guting05} provide a good reference to this
large  corpus of work.
 G\"uting {\em et al}
proposed a system of abstract data types as extensions to DBMSs to
support time-dependant geometries~\cite{Guting00}.
  Hornsby and
 Egenhofer~\cite{Hornsby02}
 introduced a  framework for modeling moving
 objects, that supports viewing objects at different
 granularities, depending on the sampling time interval.
 The possible positions of an object between two observation is
  estimated to be within two inverted half-cones that conform a
  {\em lifeline bead}, whose projection over the x-y plane is an
  ellipse.
 Another approach to moving objects studies moving objects
  on networks, basically represented as
graphs.
 Van de
 Weghe {\em et al} proposed a qualitative trajectory calculus for
objects in a GIS~\cite{Vandeweghe05}, based on the assumption that
in a GIS scenario, qualitative information is necessary (and, in
general, more useful than quantitative information).


Aggregate information is still  quite an open field, either in GIS
or in a moving objects scenario. Meratnia and de
By~\cite{Meratnia02} have tackled the topic of aggregation of
trajectories, identifying similar trajectories and merging them in
a single one, by dividing the area of study into homogeneous {\em
spatial units}.
Papadias {\em et al}~\cite{Papadias02b} index historical aggregate
information about moving objects. They aim at building a
spatio-temporal data warehouse

Regarding the addition of semantics to trajectories,  Brakatsoulas
\emph{et al}\cite{Brakatsoulas04}, in the context of trajectory
mining in road networks, propose to enrich  trajectories of moving
objects with information about the relationships between
trajectories (e.g., \emph{intersect}, \emph{meets}), and between a
trajectory and the GIS environment (s\emph{tay within},
\emph{bypass}, \emph{leave}). Extending this notion, Damiani
\emph{et al}~\cite{Damiani07} introduced the concept of stops and
moves, in order  to enrich trajectories with semantically
annotated data. With a similar idea,
 \cite{Mouza05}
propose a model where trajectories are represented by a sequence
of moves. They propose a query language based on regular
expressions, aimed at obtaining so-called mobility patterns.
 However, this language is only geared towards trajectory data,
 and does not relate trajectories with the GIS
  environment. Thus, the classes of queries addressed is limited.
 Moreover, aggregation is not considered in this language.

    We can conclude that, although  the efforts above
address particular problems,
    integrating spatial and warehousing information in a single
  framework is still in its infancy.

\section{A Data Model for Moving Objects}
\label{sec:DataModel}

 Our work is based on the data model
 introduced in \cite{Haesevoets06,Kuijpers07}. In this section we
  give an overview of this model.
  We first present the model for spatial data,
 and then we introduce the notion of moving objects.

 \subsection{Spatial Data}

  A GIS dimension
 is considered, as usual in databases, as composed  of  a schema and instances.
 Figure \ref{fig:gisdim} (left) depicts  the schema of a GIS dimension: the bottom
level of each hierarchy, denoted the {\em Algebraic part} of the
dimension, contains the infinite points in a layer, and could be
described by means of linear algebraic equalities and
inequalities~\cite{cdbook}.
 Above this part there is  the  {\em
Geometric part,} that stores the identifiers of the geometric
elements of GIS and is used to solve the geometric part of a query
(i.e. find
 the polylines -implemented as linestrings-  in a river representation).
 Each point in the Algebraic part may correspond to one or more elements
 in the Geometric part. Thus, at the \emph{GIS dimension instance} level we
 will have rollup
 \emph{relations} (denoted $r_{L}^{geom_1 \to geom_2}.$
  These relations map, for example, points in the Algebraic
 part, to geometry identifiers in the Geometric part
   For example,
 $r_{L_{city}}^{point \to Pg}(x,y,pg_1)$ says that point $(x,y)$ corresponds
 to a polygon identified by $pg_1$ in the Geometric part (note that a point may
 correspond to more than one polygon, o to more than one polylines that
 intersect with each other).

  Finally, there is the
 {\em OLAP part} of the dimension. This part contains  the
  conventional  OLAP structures, as defined in \cite{Hurtado99}.
 The levels in the geometric part are
 associated to the OLAP part via  a function,
   denoted $\alpha_{L,D}^{dimLevel \to geom}.$ For instance,
    $\alpha_{L_r,Rivers}^{riverId \to g_r}$  associates information
    about a river in the OLAP part ($riverId$), to the identifier of a
     polyline ($g_r$) in
 a  layer containing rivers ($L_r$) in the Geometric part.

\begin{figure}[t]
\begin{tabular}{l c r}
   {\psfig{figure=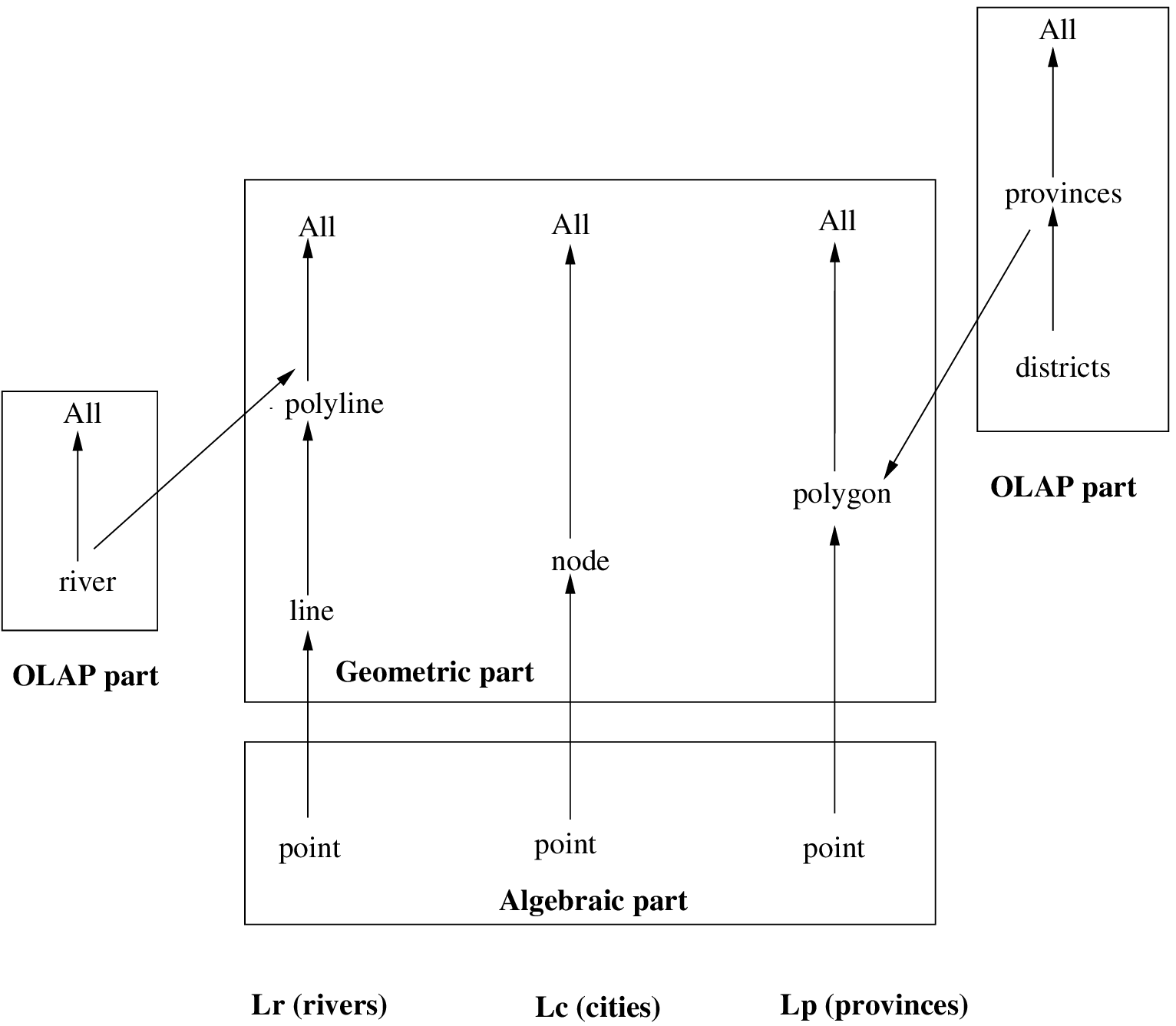,width=2.3in,height=2.0in}}
   &\phantom{josk}&
   {\psfig{figure=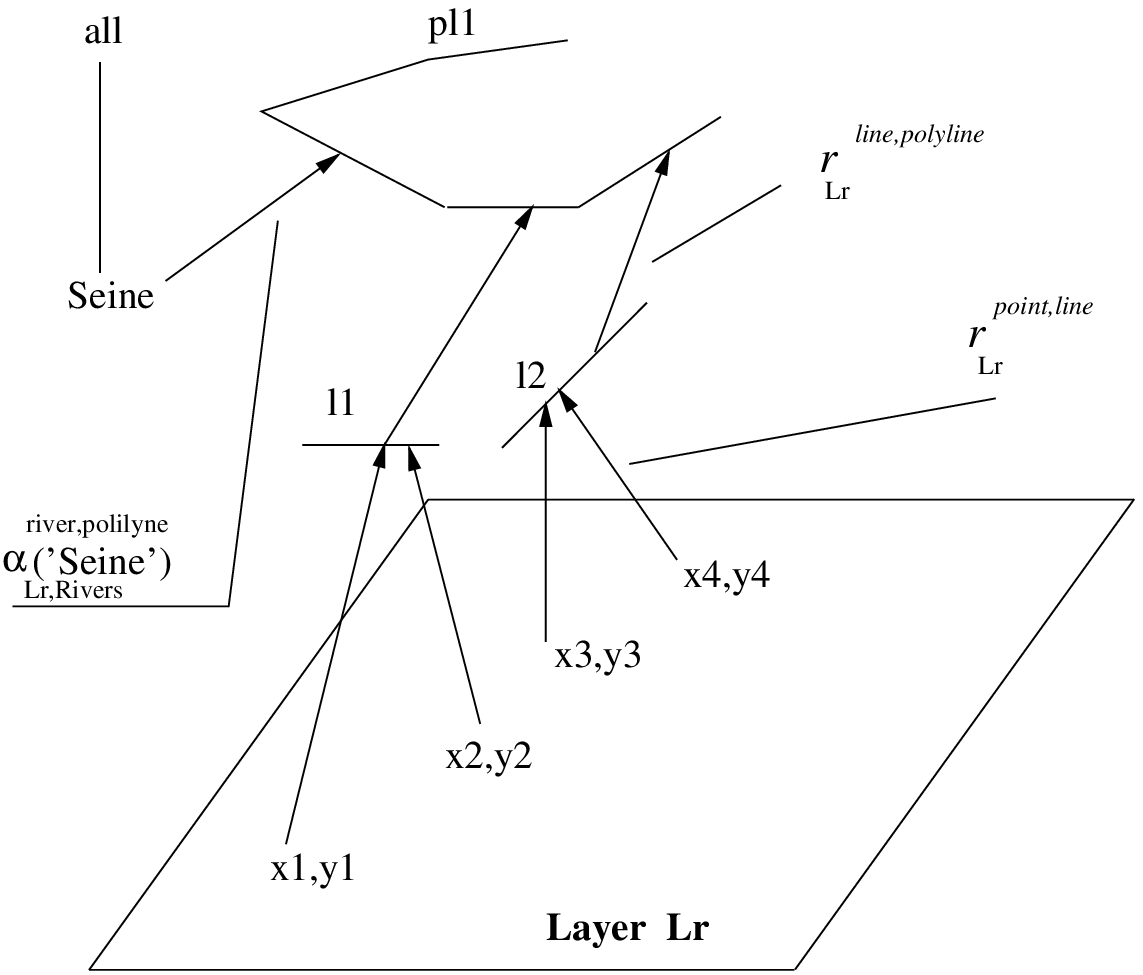,width=2.3in,height=2.0in}}
   \end{tabular}
\caption{A GIS dimension schema (left) and A GIS dimension
instance (right)}
 \label{fig:gisdim}
\end{figure}


\begin{example}\rm Figure \ref{fig:gisdim} (left) shows a  GIS dimension
schema, where we defined three layers, for rivers,  cities, and
 provinces, respectively.   The schema is composed of three
 graphs; the graph for rivers contains edges saying
 that a point $(x,y)$ in the algebraic part relates to a line identifier
   in the
 geometric part, and that in the same portion of the dimension,
 this line aggregates on a  polyline identifier.

In the OLAP part we have information given by two dimensions,
  representing  districts and rivers, associated to the
corresponding graphs, as the figure shows. For  example, a river
identifier at the bottom layer of the {\em Rivers} dimension representing
rivers in the OLAP part, is mapped to the  polyline dimension
level in the geometric part in the graph in the rivers layer $L_r$.

 Figure \ref{fig:gisdim} (right) shows a portion of a GIS dimension
 instance of the rivers layer $L_r$ in the dimension schema of
  the schema in the left of the figure.
  Here, an instance of a GIS dimension in the
 OLAP part is associated to the polyline $pl_1,$ which
 corresponds to the Seine river. For simplicity we only
   show  four different points at the
 $point$ level $\{(x_1,y_1),\ldots,(x_4,y_4)\}.$ There is a
 relation $r_{L_r}^{point,line}$ containing the association of points to the lines
 in the \emph{line} level, and a relation $r_{L_r}^{line,polyline},$
  between the line and polyline levels, in the same layer.
  \qed
\end{example}


Elements in the geometric part  can be associated with {\em
facts}, each fact being quantified by one or more {\em measures},
not necessarily  a numeric value.
 Of course, besides the GIS fact tables,  there may also be classical fact
tables in the OLAP part, defined in terms of the OLAP dimension
schemas. For instance, we could either store the population
associated to a polygon identifier,  or in a data warehouse fact table,
  with schema $(state,
Year, Population).$

 \subsection{Moving Object Data Representation}

Besides the static information representing geometric components
(i.e., the GIS), for representing time in the OLAP part there will
be a Time dimension (actually, there could be more than one Time
dimensions, supporting, for example, different notions of time).
Also, as it is well-known in OLAP, this time dimension can have
different configurations that depend on the application at hand.
  Moving objects are integrated  in the former framework using  a
distinguished fact table denoted {\em Moving  Object Fact  Table}
(MOFT).

First, we say what a trajectory is. In practice, trajectories are available
by a finite sample of $(t_i,x_i,y_i)$ points, obtained by observation.

\begin{definition}[Trajectory]\rm
A \emph{trajectory}  is a list of time-space points
$\langle (t_0, \ab x_0,\ab y_0),\ab (t_1, \ab x_1,\ab
y_1),...,(t_N, \ab x_N,\ab y_N)\rangle$, where $t_i, \ab x_i,\ab
y_i \in\R$  for $i=0,...,N$ and $t_0<t_1<\cdots <t_N$.
We call the interval $[t_0,t_N]$ the \emph{time domain} of the trajectory.
 \qed
 \end{definition}

For the sake of finite representability, we may assume that the
time-space points $(t_i,x_i,y_i)$,
  have rational coordinates.

A moving object fact table (MOFT for short, see the table in the
right hand side of Figure~\ref{fig:paris}  contains a finite
number of identified trajectories. Formally:

\begin{definition}[Moving  Object Fact  Table]\label{def:moft}\rm
Given a  finite set $\mathcal{T}$ of trajectories, a {\em Moving Object Fact Table} (MOFT) for
$\mathcal{T}$ is a relation with schema $<Oid, T, X,
Y>,$ where $Oid$ is the identifier of the moving object, $T$
represents time instants, and $X$ and $Y$ represent the spatial
coordinates of the objects. An instance $\mathcal{M}$ of the above schema contains
 a finite number of tuples of the form $(O_i, t, x, y)$, that
represent the position $(x,y)$ of the object $O_i$ at instant $t,$
for the trajectories in  $\mathcal{T}$.\qed
\end{definition}

\section{A Query Language for Aggregation of Moving Object Data}
\label{sec:FOlang}

 The aggregation  queries
that we  address in this paper are based on a first-order moving
object query language $ {\cal L}_{mo}$ and they are of the
following types:

\begin{itemize}

\item the {\sc Count} operator applied to sets of the form
$\{O_{id}\mid \phi(O_{id})\},$
where moving objects identifiers satisfying some $ {\cal L}_{mo}$-definable property $\phi$ are
collected;

\item the {\sc Count} operator applied to sets of the form
$\{(O_{id},t)\mid \phi(O_{id},t)\},$ where moving objects identifiers
combined with time moments,  satisfying some $ {\cal L}_{mo}$-definable
property $\phi$, are collected (assuming that this set is finite; otherwise the count is undefined);

\item the {\sc Count} operator applied to sets of the form
$\{(O_{id},t,x,y)\mid \phi(O_{id},t,x,y)\},$ where moving objects id's
combined with time and space coordinates, satisfying some $ {\cal L}_{mo}$-definable property $\phi$,
are collected (assuming that this set is finite);

 \item the {\sc Area}  operator applied to sets of the form
 $\{(x,y)\in \R^2\mid\phi(x,y)\},$
which define some  $ {\cal L}_{mo}$-definable  part of  the plane $\R^2$ (assuming that this set is linear and bounded);

\item the {\sc Count}, {\sc Max} and {\sc Min} operators applied to sets
of the form  $\{t\in \R\mid\phi(t)\},$
when the  $ {\cal L}_{mo}$-definable condition $\phi$ defines a finite set of time instants and the
{\sc TimeSpan} operator when $\phi$ defines an infinite, but bounded
 set of time instants (the semantics of {\sc Count}, {\sc Max} and {\sc Min} is clear and
{\sc TimeSpan} returns the difference between the maximal and minimal
moments in the set);

\item the {\sc Max-l}, {\sc Min-l}, {\sc Avg-l} and  {\sc TimeSpan-l} operators
applied to sets of the form  $\{(t_s,t_f)\in
\R^2\mid\phi(t_s,t_f)\},$ which represents an   $ {\cal L}_{mo}$-definable  set of time
intervals. The meaning of these operators is respectively  the maximum,
minimum and average  lengths of the intervals if there is  a finite number of intervals
and the  timespan of the union of these intervals in the last case;

\item the {\sc Area}  operator applied to sets of the form
$\{g_{id}\mid\phi(g_{id})\},$
where identifiers of elements of some geometry (in the geometric
part of our data model), satisfying an  $ {\cal L}_{mo}$-definable $\phi$ are
 collected. The meaning of this operator is the total area covered by the geometric elements corresponding to the identifiers.
\end{itemize}

%

Obviously, the above list is not complete, but is covers the most
interesting and usual cases (see~\cite{Kuijpers07} for an
extensive list of examples of moving object aggregation
queries). For instance, sets like $\{(t,x)\in \R^2\mid\phi(t,x)\}$
do not correspond to any obvious entity we would like to aggregate
over.

To complete the description of our moving-object aggregation language, the query language $ {\cal L}_{mo}$ remains to be defined.
In the  $ {\cal L}_{mo}$-definable sets considered above, we can see that there are variables
of different kinds, like $O_{id}, t, x, y$ and  $g_{id}.$ In fact, $ {\cal L}_{mo}$ is a multi-sorted
first-order logic using variables of
these types to define sets as considered above.
We now define $ {\cal L}_{mo}$ more formally.

\begin{definition}\label{def:FOlang}\rm

The first-order query language $ {\cal L}_{mo}$ has four types of
variables: \emph{real variables}  $x,y,t,\dots $; \emph{name variables}
$O_{id},...$; \emph{geometric identifier variables} $g_{id},...$ and
\emph{dimension level variables} $a,b,c,...,$ (which are also use for
dimension level attributes).

Besides (existential and universal) quantification over all these variables, and the usual logical
connectives $\land,\lor,\lnot...$, we consider the following
functions and relations to build atomic formulas in $ {\cal L}_{mo}$:
\begin{itemize}
 \item
 for every rollup
function in the OLAP part, we have a function symbol $f_{D_k}^{G_i \to G_j}$,
where $G_i$ and $G_j$ are geometries and $D_k$ is a dimension;
\item analogously, for every rollup relation in the GIS
part, we have a relation symbol $r_{L_k}^{G_i \to G_j}$, where $G_i$ and $G_j$ are
geometries and $L_k$ is a layer;
\item for every $\alpha$ relation associating the OLAP and   GIS
parts in some layer $L_i$, we have a relation symbol
$\alpha_{L_k,D_\ell}^{A_i \to G_j}$, where $A_i$ is a OLAP dimension level and $G_j$
is a geometry, $L_k$ is a layer and $D_\ell$ is a dimension;
\item  for every dimension
level \emph{A}, and every attribute $B$ of $A,$ denoted $A.B,$
there is a function $\beta_{D_k}^{A \to B}$ that maps elements of $A$ to elements of $B$
in dimension $D_k$;
\item we have functions, relations and constants that can be applied to the
alpha-numeric data in the OLAP part (e.g., we have the $\in$ relation to say
that an element belongs to a dimension level, we may have $<$ on income values and
the function $concat$ on string values);
\item for every MOFT,  we have a 4-ary relation $\mathcal{M}_i$;
\item we have arithmetic operations $+$ and $\times,$ the
constants $0$ and $1$, and the relation $<$ for real numbers.
\item finally, we  assume the equality relation for all types of variables.
\end{itemize}
If needed, we may also assume other constants, e.g., for
object identifiers.
\qed
\end{definition}

Definition \ref{def:FOlang} describes the syntax of the language
$  {\cal L}_{mo}$. The interpretation of all variables, functions,
relation, and constants is standard, as well as that of the
logical connectives and quantifiers. We don't define the semantics formally but
  illustrate   through an elaborate example.

\begin{example}\rm
\label{ex:FOlang}
 Let us consider the
query \emph{``Give the total number of buses per hour in the
morning in the Paris districts  with a monthly income of less than
\eu 1500,00.''}


We use the MOFT  $\mathcal{M}$  (Figure~\ref{fig:paris}, left),
that contains the moving objects samples.
For clarity, we will denote the geometry polygons by $Pg$, polylines by  $Pl$
and point by $Pt$. We use \emph{distr}
to denote the   level district in the
OLAP part of the dimension
 schema. The GIS layer which contains district information is called  $L_d$.
 As in the above definition,  we assume that the layers to which a function refers
 are implicit by the function's name.
 For instance, in the expression
 $\alpha_{L_d,Distr}^{distr,Pg}(n)=p_g$,
 the district variable  $n$ is mapped to the polygon with variable name $p_g$ that
  is in the layer $L_d$, indicated by the function $\alpha_{L_d,Distr}^{distr,Pg}$
  (here $Distr$ is a dimension in the OLAP part representing districts).
 Thus, the result of the query returning the region with
 the required income is expressed as:
$$\{(x,y)\mid \exists n \exists g_1
(r_{L_d}^{Pt \to Pg}(x,y,g_1)~\land~
\alpha_{L_d,Distr}^{distr,Pg}(n)=g_1~\land~\beta_{Distr}^{distr\to income} (n)  < 1.500)\}.$$

In this expression, $r_{L_d}^{Pt \to Pg}(x,y,g_1)$ relates points to
polygons in the district layer; the
 function $\alpha_{L_d,Distr}^{distr,Pg}(n)=g_1$ maps
 the district identifier $n$ in the OLAP part to the geometry
 identifier $g_1$; and $\beta_{Distr}^{distr\to income} (n) $ maps
 the district identifier $n$ to the value of the income attribute which
 then is compared by an OLAP relation $<$ with a OLAP constant $1.500$.

 The instants
corresponding to the morning hours mentioned in the fact tables
are obtained through the rollup functions in the Time dimension.
We assume in the Time dimension a category denoted {\em
timeOfDay}, rolling up to the dimension category {\em hour} (i.e.,
$timeOfDay~\to~hour$).  The aggregation of the values in the  fact
table corresponding only to morning hours is computed with the
following expression: $\mathcal{M}_{morning}= \{ (Oid,t,x,y)\mid
f_{Time}^{timeOfDay\to hour}(t)=\mbox{``Morning''}~\land~
\mathcal{M}(Oid,t,x,y)\}.$ In this formula $\mbox{``Morning''}$
appears as a constant related to the OLAP part. Finally, the query
we discuss reads:

$$\displaylines{\quad \mbox{\sc Count}\{(Oid,t)\mid \exists x
\exists y \exists g_1 \exists n (n \in distr\land
\mathcal{M}_{morning}(Oid,t,x,y)~\land~\hfill{}\cr\hfill{}
~r_{L_d}^{Pt \to Pg}(x,y,g_1)~\land~
\alpha_{L_d,Distr}^{distr, Pg}(n)=g~\land~\beta_{Distr}^{distr\to income} (n) < 1,500)\}.\quad}$$


If we would change the given aggregation query to ``Give the total
number of buses per hour in the morning within 3 km from a  Paris
district with a monthly income of less than \eu 1500,00.'' then we
would need $+$, $\times$ and $<$ to express the distance
constraint. This would introduce a quadratic polynomial in the
formula to express that some points are less than 3 km apart. This
concludes the example.\qed
\end{example}


\begin{proposition}
Moving object queries expressible in  ${\cal L}_{st}$ are computable. The
proposed aggregation operators are also computable.
\end{proposition}

\begin{proof} (Sketch)
The semantics of ${\cal L}_{st}$ expressions is straightforward apart from the
subexpressions that involve $+$, $\times$ and $<$ on real numbers and quantification
over real numbers. These subexpressions belong to the formalism of constraint databases
and they can be evaluated by quantifier elimination techniques~\cite{cdbook}.

The restrictions that we imposed on the applicability of the aggregation operators
make sure that they can be effectively  evaluated. In particular, the area of a set
 $\{(x,y)\in \R^2\mid\phi(x,y)\}$ is computable when this set is semi-linear and
 bounded. This area can be obtained by triangulating such linear sets and adding the areas of the triangles.\qed
\end{proof}

\section{Stops and Moves of Trajectories}
\label{sec:trajectories}

In this section, we define what the stops and moves of a
trajectory are. In a GIS scenario, this definition is dependent on the particular
places of interest in a particular
application.
For instance, in a tourist application, places of interest may be hotels, museums and
churches. In a traffic application, places of interest may be road segments, road junctions
and traffic lights.
First, we define the notion of
``places of interest of an application'' (PIA).


\begin{definition}\rm[Places of Interest]
A \emph{place of interest (PoI)} $C$ is a tuple $(R_C,\Delta_C)$, where
$R_C$ is a (topologically closed) polygon, polyline or point in
$\R^2$  and $\Delta_C$ is a strictly positive real number. The set
$R_C$ is called the \emph{geometry} of the PoI $C$ and
$\Delta_C$ is called its \emph{minimum duration}.

The \emph{places of interest  of an application} (PIA) $\mathcal {P}$ is a finite collection
of PoIs with mutually disjoint geometries.\qed
  \end{definition}

 \begin{definition}\rm [Stops and moves  of a trajectory]
Let $T=\langle (t_0, \ab x_0,\ab y_0),\ab (t_1, \ab x_1,\ab
y_1),...,(t_n, \ab x_n,\ab y_n)\rangle$ be a trajectory
and let ${\mathcal
{P}}=\{C_1=(R_{C_1},\Delta_{C_1}),\ab ...,\ab
C_N=(R_{C_N},\Delta_{C_N})\}$ be a PIA.

A \emph{stop of $T$ with respect to $\mathcal {P}$} is a maximal
contiguous subtrajectory $\langle (t_i, \ab x_i,\ab y_i),\ab
(t_{i+1}, \ab x_{i+1},\ab y_{i+1}),...,(t_{i+\ell}, \ab
x_{i+\ell},\ab y_{i+\ell})\rangle$ of $T$ such that for some
$k\in\{1,...,N\}$ the following holds: (a) $(x_{i+j},y_{i+j})\in
R_{C_k}$ for $j=0,1,...,\ell$; (b) $t_{i+\ell}-t_i>\Delta_{C_k}$.


A \emph{move  of $T$ with respect to $\mathcal {P}$} is: (a) a
maximal contiguous subtrajectory of $T$ in between two temporally
consecutive stops of $T$; (b)  maximal contiguous subtrajectory of
$T$ in between the starting point of $T$ and the first stop of
$T$; (c) a maximal contiguous subtrajectory of $T$ in between the
last stop of $T$ and ending  point of $T$; (d) the trajectory $T$
itself, if $T$ has no stops. \qed



\end{definition}

  \begin{figure}[t]
 \centering
 \input{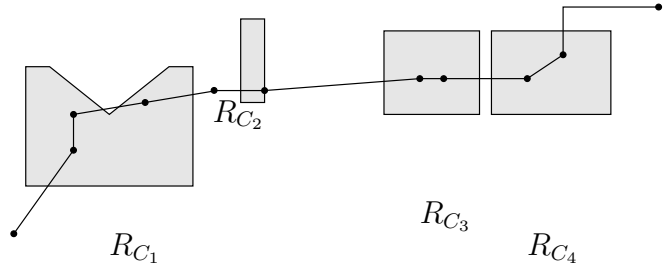}
 \caption{An example of a trajectory with two stops and three moves.}\label{3stops}
 \end{figure}

 Figure~\ref{3stops} illustrates these concepts.
 In this example, there are four places of interest with geometries $R_{C_1}, R_{C_2}, R_{C_3}$ and $R_{C_4}$.
 The trajectory $T$ is depicted here by linearly interpolating between its sample points, to indicate their order.
 Let us imagine that  $T$
 is run through from left to right.
 If the three sample points in $R_{C_1}$
 are temporally far enough apart (longer than $\Delta_{C_1}$), they form a stop.
 Imagine that further on, only the two sample points in
 $R_{C_4}$ are temporally far enough apart to form a stop. Then we have two stops
 in this example and three moves.

 We remark that our definition of stops and moves of a trajectory is arbitrary and can
  be modified in many ways. For example, if we would work with linear interpolation of
  trajectory samples, rather than with samples, we see in Figure~\ref{3stops}, that the
 trajectory briefly leaves $R_{C_1}$ (not in a sample point, but in the interpolation).
 We could incorporate a tolerance for this kind of small exists from PoIs in the definition,
 if we define stops and moves in terms of continuous trajectories, rather than on terms of
 samples.  The following property is straightforward.

  \begin{proposition}\label{prop:pointquery}
  There is an algorithm that returns, for any input $({\cal P},T)$
  with $\cal P$  a PIA and $T$ a trajectory
  $\langle (t_0, \ab x_0,\ab y_0),\ab (t_1, \ab x_1,\ab
y_1),...,(t_n, \ab x_n,\ab y_n)\rangle$, the stops of $T$ with
respect to $\cal P$. This algorithm  works in time $\ord{n\cdot p}$,
 where $p$ is the complexity of answering the point-query~\cite{Rigaux02}.
\qed
\end{proposition}

\section{A Stops and Move Fact Table}\label{sec:stopsmoves}

Let the places of  interest  (PoIs) of an application (PIA) be given.
In this section, we describe how we go from MOFTs to application dependent
compressed MOFTS, where $(O_{id},t_i,x_i,y_i)$ tuples are replaced by
$(O_{id},g_{id},t_s,t_f)$ tuples.
In the latter tuples, $O_{id}$ is a moving object
 identifier,  $g_{id}$ is an identifier
of the geometry of a place of interest and $t_s$ and $t_f$ are two
time moments that encode the time interval $[t_s,t_f]$ of a stop.
The idea is to replace the trajectories in a MOFT that are stored
there as  samples,
 by a
\emph{stops MOFT} that represents the  same trajectory
 more concisely by listing its stops and the time intervals spent in the stops.

 In our model, application information about the PoIs is
 stored in the OLAP part
 as OLAP dimensions.
 For example, if  hotels are places of interest, we will need to
  create a dimension Hotels such that its bottom level contains the
  identifier for the hotels and some hierarchy that is
  specific for hotels, e.g., a hotel may belong to the 3-star category.
  Given that these dimensions depend on a particular application,
  we define, at the conceptual level, a \emph{Generic} virtual
  dimension, from which different dimensions can be generated.

To start with, we assume that the places of interest are stored in
moving object OLAP in the following way: the geometries of the
PoIs  are represented in a layer
  in the GIS denoted $L_{PoI}$ (e.g., a layer containing polygons
   that represent hotels or
   a layer containing polylines that represent street segments).
   Data describing the places of interest is stored in the OLAP
   part.

Figure~\ref{fig:schemanew} illustrates how the information about
the places of interest is represented in our model. In this
figure,  in the OLAP part there is a virtual dimension, which we
call the \emph{Generic PoI}, that will be instantiated by as many
types of places of interest as a particular application requires
(in the figure, we show an instantiation for hotels). The bottom
level of this dimension is denoted $PoI_b.$ There is also a
function that maps  the bottom level of the instances of the
Generic PoI (GPoI) to geometries in the geometric part, in the
layer corresponding to the PoIs, denoted $L_{PoI}$. In
Figure~\ref{fig:schemanew}, hotelId is mapped to the geometry
Polygon in the layer $L_{PoI}$. The minimum duration of a PoI is
stored as an attribute of the bottom level of the instances of the
GPoI. For example, an attribute of level hotelId in Figure
~\ref{fig:schemanew}.
   At the  instance level, analogous to what we explained in
  Section~\ref{sec:DataModel}, the function
  $\alpha_{L_{pPoI},D}^{P_i \to G_i}$ maps
 elements in the bottom level ($P_i$) of the instances of the GPoI,
  to the geometric identifiers of the places of interest in the
  geometric part (in Figure \ref{fig:schemanew}, the function is
  defined as $\alpha_{L_{pPoI},Hotels}^{hotelId \to Polygon}$).

\begin{figure}[t]
\centerline{\psfig{figure=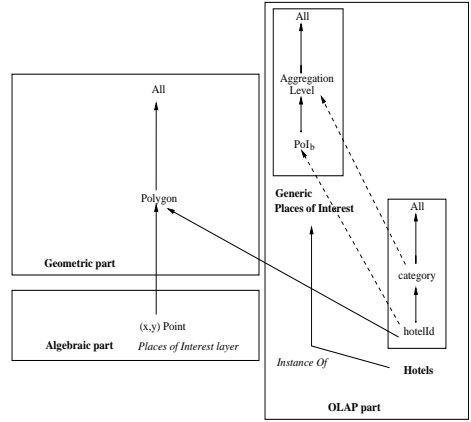,width=2.4in,height=2.2in}}
\caption{Adding Stops to the Data Model} \label{fig:schemanew}.
\end{figure}

\begin{definition}[{\cal SM}-MOFT]\rm \label{def:encmoft}
Let ${\cal
P}=\{C_1=(R_{C_1},\Delta_{C_1}),...,C_N=(R_{C_N},\Delta_{C_N})\}$
be a PIA of PoIs and let $\cal M$ be a MOFT. The \emph{{\cal
SM}-MOFT ${\cal M}^{sm}$ of $\cal M$ with respect to $\cal P$}
consist of the tuples $(O_{id},g_{id},t_s,t_f)$ such that (a)
$O_{id}$ is the identifier of a trajectory in $\cal M$; (b)
$g_{id}$ is the identifier of the geometry of a PoI
$C_k=(R_{C_k},\Delta_{C_k})$ of $\cal P$ such that the trajectory
with identifier $O_{id}$ in $\cal M$ has a stop in this PoI during
the time interval $[t_s,t_f]$. This interval is called the
\emph{stop interval} of this stop. \qed


 \end{definition}

The table in Figure~\ref{fig:parisEnc} (left) gives an example of
a {\cal SM}-MOFT. The following property shows that {\cal
SM}-MOFTs can be
 defined in the moving object query language ${\cal L}_{mo}$.

  \begin{proposition}\label{prop:smmoftinfo}
  There is an ${\cal L}_{mo}$ formula $\phi_{sm}(O_{id},g_{id},t_s,t_f)$ that defines the
 {\cal SM}-MOFT ${\cal M}^{sm}$ of $\cal M$ with respect to $\cal P$. \qed
\end{proposition}

 We omit the proof of this property but remark that the use of the formula
 $\phi_{sm}(O_{id},g_{id},t_s,t_f)$ allows us to speak about stops
 and moves of trajectories in ${\cal L}_{mo}$. We can therefore add predicates
 to define stops and moves of trajectories as syntactic sugar to ${\cal L}_{mo}$.

\section{A Query Language for Moving Objects}
\label{sec:language}

In this section we will show how the language $\mathcal{L}_{mo}$
and the model supporting it, can yield sub-languages that can
address many interesting  aggregation queries for moving objets in
a GIS environment. We will sketch a query language based on path
regular expressions, along the lines proposed by \cite{Mouza05}.
However, our language goes far beyond, taking advantage of the
integration between GIS, OLAP and moving objects provided with our
model. Moreover, queries that do not require access to the MOFT
can be evaluated very efficiently, making use of the SM-MOFT.

The idea is based on the construction of a graph representing the
stops and moves of a \emph{single trajectory} as follows: from the
SM-MOFT $\mathcal{M}^{sm}$ we construct a graph $\mathrm{G}$ as
follows. For each different $g_{id}$ in $\mathcal{M}^{sm},$ there
is a node $v$ in $\mathrm{G},$ denoted $v(g_{id}),$ which is
assigned a unique node number $n$. Further, there is an edge $m$
in $\mathrm{G}$ between two nodes
  $v(g_{id_1})$ and $v(g_{id_2}),$
  for every pair of $t_1,t_2$ of consecutive tuples in
  $\mathcal{M}^{sm}$ with the same $O_{id}.$
   Each node $v$ is augmented
   with two functions  and one set: (a) the function $extent(v)$ returns  the
   identifier $p_{id}$ of the PoI in the OLAP part of the
   model (i.e., the  $p_{id}$ such that $\alpha_{L_{PoI},D}^{P_i \to G_i}(p_{id})=g_{id}$);
    (b) the function $label(v)$ returns  the dimension in the
    OLAP part to
   which a given PoI $p_i$ belongs (v.g, Hotels, Museums, and so on);
    (c) a set of Stop Intervals (technically a temporal element) \emph{STE(v)},
     containing the  stop intervals of the object at $v.$
      Note that an object may be
      at a stop more than one time within a trajectory. Further,
      these is an ordered set, given that the intervals are
      disjoint by definition and consecutive by construction.
      We denote the graph constructed in this way an SM-Graph.

\begin{example}
Let us consider the  SM-MOFT table $\mathrm{M}^{sm}$ based on
 the SM-MOFT of Figure \ref{fig:parisEnc} (left). We will use the
 SM-Graph for the trajectory  such that $O_{id}=O2,$ obtained as
 $\sigma_{O_{id}=O_2}(\mathcal{M}^{sm}).$
 Also, we will denote in our examples,  $H_i,$ $M_i,$ and $T_i,$
 the instances of hotels, museums and tourist attractions, respectively.
 Figure  \ref{fig:parisEnc} (right) shows the  SM-Graph. \qed
\end{example}



\begin{figure}[t]
\centerline{
\begin{tabular}{l c r}
   {\scriptsize \begin{tabular}{|c|c|c|c|}
        \hline
       $O_{id}$ & $g_{id}$ & $t_s$& $t_f$ \\\hline
   $O_1$ & $H_1$ & $0$&$10$  \\
    $O_1$ & $L$ & $20$&$30$  \\
 $O_1$ & $H_1$ & $100$&$140$  \\
    $O_2$ & $H_2$ & $0$&$1$  \\
    $O_2$ & $L$ & $25$&$40$  \\
    $O_2$ & $E$ & $50$&$80$  \\
    $O_2$ & $H_2$ & $120$&$140$  \\
    $O_3$ & $H_2$& $0$&$10$  \\
    $O_3$ & $E$& $10$&$40$  \\
$O_3$ & $H_2$& $60$&$140$  \\ \hline
        \end{tabular}}
   &\phantom{joske}&
\raisebox{-2cm}{\psfig{figure=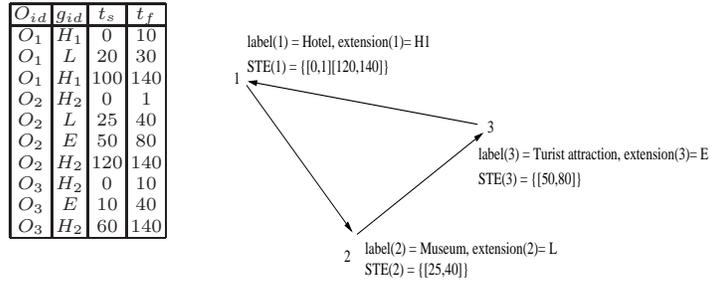,width=2.5in,height=1.3in}}
        \end{tabular}
} \caption{An SM-MOFT for the running example (left); An SM-Graph
(right) for this table.} \label{fig:parisEnc}
\end{figure}

\def\iinter{{\sqcap}}
\def\iunion{{\sqcup}}
\def\overr{{\triangleleft}}
\def\iscl{{\tt iscl}}

%

We will also need some  operators on time intervals.
  We say that an interval $I_1=[t_1,t_2]$ strictly
precedes $I_2=[t_3,t_4],$ denoted $I_1 \lessdot I_2,$ if $t_1
<t_2<t_3 <t_4.$ We also say that $t \vartriangleleft [t_1,t_2] $
returns \emph{True} if $t_1 < t< t_2.$
 Note that all stop intervals $I_1, I_2$ of the same trajectory
 are such that either $I_1 \lessdot I_2$ or $I_2 \lessdot I_1.$



Now we are ready to define a simplified query language for moving
object aggregation, taking advantage of the concept of stops and
moves, but powerful enough to combine (to some extent) the notion
of regular path expressions and first order constraints.
 We assume that MOFTs are
well-defined, thus the graphs are temporally consistent. In
addition,  each edge in an SM-Graph is univocally defined by the
intervals of the stop temporal elements of  the beginning and
ending nodes of the edge. In other words, if there exist two edges
from a node $v_1$ to a node $v_2.$ Each node must have associated
two stop temporal intervals, $STE(v_1)=\{I_1,I_3\}$ and $STE(v_2)
=\{I_2,I_4\},$
 where $I_1 \lessdot I_2 \lessdot I_3 \lessdot I_4$ holds.

 A  first observation at the definition of the
  $\mathcal{SM}$-Graph $\mathcal{G}$
  reveals  that the  graph  can be seen as a
 DFA accepting  regular
 expressions over the labels of the nodes in the graph.
  This becomes clear if, in the graph of Figure  \ref{fig:parisEnc} (right)
  we replace $v$ by $label(v)$. In this case, the nodes
  labeled   $M_i,$ and  $H_i$ will become $M,$ and $H,$
   respectively (shorthand for \emph{Museums} and \emph{Hotels}.
   We call this graph ASM-Graph (the A stands for aggregation).
   As a second observation,  we can
  think on a language such that the DFA accepting this language is
 contained in  the ASM-Graph. Thus, a trajectory
 satisfies a query $Q$ if the DFA of the query is contained
 in $G.$

 \begin{definition}[Regular Expressions Language for Stops and Moves]\rm
  \label{def:axpre}
 An \emph{regular expression on stops and moves,}
 denoted \emph{RESM} is an expression generated
 by the grammar

\[ E \longleftarrow   dim \mid dim[cond]
    \mid  (E)^* \mid E.E \mid \epsilon \mid ? \]

where $dim \in D$ (a set of dimension names in the OLAP part),
$\epsilon$ is the symbol representing the empty expression, ``.''
means concatenation, and $cond$ represents a condition over
$\mathcal{L}_{st}$. The term ``?'' is a wildcard meaning  ``any
sequence of any number of $dim$''.\qed
\end{definition}

The aggregate language  is built on top of RESMs:  for each
trajectory $T$ in an SM-MOFT such that there is a sub-trajectory
of $T$  that matches the RESM, the query returns the $O_{id}$ of
$T.$
Then, we
can apply the aggregate function {\sc Count} to the set returned.


%
%

%
%
%
%
%
%
%

We explain the semantics of RESM-based language using the query:
\emph{``total number of trajectories that went from a ``Hilton''
hotel to a tourist attraction, stopping at a museum.''}, whose
RESM reads: {\sc Count}(H[name=`Hilton'].?.M.?.T).

Note that ``name'' is an attribute of the PoI identifier $p_{id}$
 in the OLAP part (an attribute of the extension of the node).
Then, for each trajectory, and for each instantiation with a value
H, M or T, of a node in the graph,  the variable $name$ is
instantiated with the value $v_i$ corresponding to the attribute
$name$ of $p_{id}$  in the OLAP part such that
$extension(v).name=v_i$ in the dimension $D=Hotel.$ The condition
on the node is then checked.   Finally, if there is a
sub-trajectory matching the RESM, then its $O_{id}$ counts for the
aggregation.

As another example, the query \emph{``total number of trajectories that
went from a Hilton hotel to the Louvre, in the morning.''} \\
{\sc Count}(H[name='Hilton'].?.M[name='Louvre' $\land \exists~t
\vartriangleleft I \land f_{Time}^{timeId\to
TimeOfDay}(t)=``morning''$ ])

The semantics of the first condition is analogous to the semantics
of the query above. The same occurs with
 the condition  over $name$ in M. For the last part of the
 condition over M,   for each trajectory, and
each  instantiation of a node in the graph with a value $H$ or
$M$, $I$ is instantiated
 with values of $STE(v).$

\begin{proposition} The language defined above
is a subset of $\mathcal{L}_{mo}.$\qed
\end{proposition}

\begin{proof} (Sketch)
The proof is built on the property that, for each trajectory in an
SM-MOFT  the SM-Graph can be unfolded, and transformed into a
sequence of nodes, given that for all nodes $v$ in the graph, all
intervals in $STE(v)$ are disjoint. Thus, this sequence can then
be queried using any FO language with time variables, like
$\mathcal{L}_{mo}$\qed
\end{proof}

\section{Future Work}
\label{sec:Conclusions} Our future work will be focused in the
implementation of the model and  query languages proposed here,
and its integration with the framework introduced
in~\cite{Haesevoets06}. We also believe that the RESM language is
promising for mining trajectory data, specifically in the context
of sequential patterns mining with constraints, and we will work
in this direction.

\end{document}